\documentclass[a4paper,reqno,12pt]{amsart}

\usepackage{amsmath,amssymb,amsthm,amsaddr}
\usepackage{mathrsfs,esint,mathabx}

\usepackage[english]{babel}
\usepackage[utf8]{inputenc}
\usepackage{enumerate}
\usepackage[shortlabels]{enumitem}
\usepackage[backref]{hyperref}

\usepackage{geometry}
\usepackage{tikz-cd}
\usepackage{xcolor}

\theoremstyle{definition}
\newtheorem{definition}{Definition}[section]

\theoremstyle{plain}
\newtheorem{Theorem}[definition]{Theorem}
\newtheorem{Proposition}[definition]{Proposition}
\newtheorem{Lemma}[definition]{Lemma}
\newtheorem{Corollary}[definition]{Corollary}

\theoremstyle{remark}

\newcommand{\pt}{\partial}
\newcommand{\bs}{\backslash}
\newcommand{\rx}{\rho_X}
\newcommand{\R}{\mathbb{R}}
\let\epsilon\varepsilon
\let\phi\varphi

\newcommand{\la}{\langle}
\newcommand{\ra}{\rangle}
\newcommand{\vol}{\mathrm{Vol}}

\usepackage{csquotes}
\usepackage{multicol}    
\usepackage{enumitem}

\title{A note on the Gannon-Lee theorem}
\author{Benedict Schinnerl}
\email{benedict.schinnerl@univie.ac.at}
\author{Roland Steinbauer}
\email{roland.steinbauer@univie.ac.at}
\address[A1,A2]{University of Vienna, Faculty of Mathematics\\Oskar-Morgenstern-Platz 
	1, A-1090 Wien, Austria}
\date{\today}

\begin{document}

\begin{abstract} We prove a Gannon-Lee theorem for non-globally hyperbolic Lo\-rentzian metrics of regularity $C^1$, the most general regularity class currently available in the context of the classical singularity theorems. Along the way we also prove that any maximizing causal curve in a $C^1$-spacetime is a geodesic and hence of $C^2$-regularity. 
\end{abstract} 

\maketitle

\section{Introduction}

In the mid 1970-ies, several years after the appearance of the singularity theorems
of Penrose and Hawking (see e.g.\ \cite[Ch.\ 8]{HawEll73}),  D.\ Gannon \cite{Gan75,Gan76} and C.W.\ Lee \cite{Lee76} independently derived a body of results that relate the singularities  of a Lorentzian manifold to its topology. More precisely, in these results often dubbed Gannon-Lee theorem(s), they established that under an appropriate asymptotic flatness assumption, a non-trivial fundamental group of a (partial) Cauchy surface $\Sigma$ necessarily leads to the existence of incomplete causal geodesics. 

Most of these early results assumed global hyperbolicity, with the notable exception of \cite[Thms.\ 2.1-2]{Gan75}. However, the proofs relied on the false deduction that maximizing geodesics in a covering spacetime project to maximizing geodesics of the base, as pointed out in \cite{Gal83}. In the same paper Galloway proved a Gannon-Lee theorem without assuming global hyperbolicity by invoking the Hawking-Penrose theorem and heavily using a result from geometric measure theory. The latter fact was also reflected in  the formulation of the main theorem, which assumed an extrinsic condition on the three surface $\Sigma$ and the topological condition was that $\Sigma$ is not a handlebody\footnote{Under these conditions \cite{MSY82} guarantees the existence of a trapped surface within $\Sigma$.}. A related recent result for the globally hyperbolic case in a setting compatible with a positive cosmological constant was given in \cite{GL18}, providing a precise connection between the topology of a future expanding compact Cauchy surface and the existence of past singularities.
 
Another issue with the earlier results was that the nontrivial topology was confined to a compact region of a (partial) Cauchy surface bounded by a topological $2$-sphere $S$. In the context of topological censorship \cite{FSW95}, $S$ is naturally interpreted as a section of an event horizon and in four dimensions the $S^2$ topology is only natural in the light of Hawking's black hole topology theorem \cite[Sec.\ 9.2]{HawEll73}. However, since the latter fails to hold in higher dimensions, with more complicated horizon topologies  occurring  (see \cite{Emp08}, but \cite{GalSch06} for corresponding restrictions), the demand for higher dimensional Gannon-Lee type results with more natural assumptions on the topology of $S$ arises. Such a result was indeed given by Costa e Silva in \cite{Sil10} which also avoided the assumption of global hyperbolicity. 
More precisely the result was given for causally simple spacetimes, i.e.\ causal spacetimes where the causality relation is closed. However, we have recently discovered that this proof relies on an analogous false deduction: It uses that causal simplicity lifts to coverings, which is not true in general, as was explicitly shown in \cite{MinSil20}. In the latter paper Minguzzi and Costa e Silva also present a corrected result, which replaces the condition of causal simplicity by the assumption of past reflectivity, which also has to be supposed for certain covering spacetimes. The line of arguments using past reflectivity in the context of the Penrose singularity theorem was put forward in \cite{Min19a}, which also argues that this condition holds true in a black hole evaporation scenario. In the causality part of our arguments we will follow this path\footnote{
In an earlier version of this manuscript we aimed at a different strategy, replacing causal simplicity by null pseudoconvexity which \emph{does} lift to Lorentzian coverings. However, we have learned from Ettore Minguzzi that a crucial step in our argument, securing that maximal null pseudoconvexity plus strong causality implies causal simplicity, see also \cite[Thm.\ 2]{VMPRE19}, is false. An explicit counterexample was later given in \cite{HMS:21}. However, we do not know whether  null pseudoconvexity plus strong causality implies causal simplicity, a fact which would allow for an argument that allows to avoid assumptions on the causality of the cover.}.
\medskip

In this paper we extend the validity of the Gannon-Lee theorems to Lorentzian metrics of low regularity, in particlular to (certain) $C^1$-spacetimes. Having low regularity singularity theorems at hand is especially favorable in the context of extending spacetimes, as already noted in \cite[Ch.\ 8]{HawEll73}. In particular, they rule out the possibility to extend the spacetime to a complete one, even under mild regularity assumptions on the metric. Related results on $C^0$- and Lipschitz non-exendability have recently been given in \cite{Sbie18,GSL18,CK18,Sbie20}.

Indeed during the last couple of years the classical singularity theorems of Penrose, Hawking and Hawking-Penrose have been extended to $C^{1,1}$-regularity in \cite{KSSV15}, \cite{KSV15}, and \cite{GGKS18}. These results built upon extensions of Lorentzian causality theory to low regularity \cite{CG12,KSSV14,Min15,Sam16}. Most recently, Graf in \cite{Graf19} was able to further lower the regularity assumptions for the Penrose and the Hawking theorem to $C^1$. We will extend her recent techniques, in particular to the non-globally hyperbolic setting, to prove our result. 
\medskip

This note is organized in the following way: In the next section we discuss preliminaries for our work, especially the intricacies arising in $C^1$-regularity and we state our results. In section \ref{sec:c1} we lay the analytical foundations of the proof of the main theorem. In particular, we will show that any causal maximizer in a $C^1$-spacetime is a geodesic and hence a $C^2$-curve. Finally in section \ref{sec:proof} we provide new focusing results for null geodesics and collect all our results together to prove the main theorem.

\section{Preliminaries and results}\label{sec:results}

We begin by introducing our main notations and conventions, as well as some basic notions that are necessary to give a precise formulation of our results. Our main reference for all matters of Lorentzian geometry is \cite{ON83}. 

We will assume that all manifolds $M$ are smooth, Hausdorff, as well as second countable and of dimension $n$ with $n\geq 3$. We will consider Lorentzian metrics $g$ on $M$ that are of regularity at least $C^1$ with signature $(-,+,\dots,+)$. Another important regularity class is $g\in C^{1,1}$ which means that $g$ is $C^1$ and its first derivatives are locally Lipschitz continuous. A spacetime is a Lorentzian manifold with a time orientation, which we assume to be induced by a smooth timelike vector field.  

A curve $\gamma:I\to M$ defined on some interval $I$ is called timelike, causal, null, future or past directed, or spacelike, if $\gamma$ is locally Lipschitz continuous and its velocity vector $\dot\gamma$ which exists Lebesgue almost everywhere has the respective property. We denote the timelike and causal relation by $p\ll q$ and $p \leq q$, respectively and write 
$I^+(A)$ and $J^+(A)$ for the chronological and causal future of a set $A\subseteq M$. Finally we denote  the future horismos of $A$ by $E^+(A):=J^+(A)\bs I^+(A)$. The respective past versions of these sets will be denoted by $I^-$, $J^-$ and $E^-$, respectively. When we refer to these sets with regard to a particular metric $g$, it will appear in subscript, e.g.\ $E_g^+(A)$ denotes the future horismos of $A$ w.r.t.\ $g$. 

For Lorentzian metrics $g_1$, $g_2$ one says that $g_1$ has narrower lightcones than $g_2$ (respectively, $g_2$ has wider lightcones than $g_1$), denoted as $g_1 \prec g_2$, if $g_1(X,X)\leq 0$ implies $g_2(X,X)<0$ for any $0\not=X\in TM$. 

Throughout we will fix a complete Riemannian background metric $h$ and use its induced norm $\|\ \|_h$ and distance $d_h$. All local estimates will be independent of the choice of $h$.

We will denote the fundamental group of a manifold $M$ by $\pi_1(M)$. Also if $i:N \rightarrow M$ is a continuous map, the induced homomorphism of the fundamental groups is denoted by $i_\#: \pi_1(N) \to \pi_1(M)$.

\subsection{Low regularity} 

During the last couple of years the bulk of Lorentzian causality theory has been transferred to $C^{1,1}$-spacetimes, where the exponential map and convex neighbourhoods are still available. While convexity fails below that regularity \cite{HW51,SS18}, nevertheless most aspects of causality theory can be maintained even under Lipschitz regularity of the metric. Further below some significant changes occur \cite{CG12,KGSS20}, while some robust features continue to hold even in more general settings \cite{Min18,KS18,BS18,GKSS19}. In particular, for $C^1$-spacetimes the push-up principle is still valid and $I^+(A)$ is open for any set $A\subseteq M$.

However, there are two essential features one loses when going from the $C^{1,1}$-setting to $C^1$-spacetimes:  uniqueness of solutions of the geodesic equation and the local boundedness of the curvature tensor. Given the first fact, one has to make a choice concerning the definition of geodesic completeness. We will follow the natural approach of \cite{Graf19} and define a spacetime to be timelike (respectively null or causal) geodesically complete if \emph{all}\footnote{The alternative would be to only demand the existence of \emph{one} complete geodesic for every timelike (or null or causal) initial condition to the geodesic equation.} inextendible timelike (respectively null or causal) geodesics are defined on $\mathbb{R}$. 

Concerning the second issue, first note that $C^1$ is well within the maximal class of spacetimes allowing for a (stable definition of) distributional curvature, which is $g$ locally in $H^1\cap L^\infty$ \cite{GT87,LM07,SV09}. The Riemann and the Ricci tensor are then tensor distributions in $\mathcal{D}'\mathcal{T}_{3}^{1}(M)$ and $\mathcal{D}'\mathcal{T}_{2}^{0}(M)$, respectively, where we recall that 
\begin{equation}\label{eq:3}
 \mathcal{D}'\mathcal{T}_{s}^{r}(M):=\Gamma_{c}\left(M,T_{r}^{s}M\otimes\vol(M)\right)'
 =\mathcal{D}'\left(M\right)\otimes_{\,\mathcal{C}^{\infty}}\mathcal{T}_{s}^{r}\left(M\right)\,.
\end{equation}
Here $\vol(M)$ is the volume bundle over $M$, $\Gamma_c$ denotes spaces of sections with compact support and $\mathcal{D}'\left(M\right)$ is the space of scalar distributions on $M$, i.e.\  the topological dual of the space of compactly supported volume densities $\Gamma_c(M,\vol(M))$. 
We remark that, containing derivatives of the continuous connection, the distributional Riemann and Ricci tensors are of order one and that the usual coordinate formulae 
\begin{align}
 \mathrm{Riem}^m_{\;\;\;ijk} &:= \partial_j \Gamma^m_{ik}-\partial_k \Gamma^m_{ij}+\Gamma^m_{js}\Gamma^s_{ik}-\Gamma^m_{ks}\Gamma^s_{ij}\,,\\
 \mathrm{Ric}_{ij} &:= \partial_m \Gamma^m_{ij}-\partial_j \Gamma^m_{im}+\Gamma^m_{ij}\Gamma^k_{km}-\Gamma^m_{ik}\Gamma^k_{jm}
\end{align}
apply. For further details on tensor distributions see \cite[Ch.\ 3.1]{GKOS01}.
\medskip

Naturally, we define curvature bounds resp.\ energy conditions using the notion of positivity for distributions, i.e.\  $\mathcal{D}'\left(M\right)\ni u\geq 0$ ($u>0$) if 
$\langle u,\mu\rangle\geq 0$ ($>0$) for all non-negative (positive) volume densities $\mu\in \Gamma_c(M,\vol(M))$. Then $(M,g)$ is said to satisfy the strong energy condition (resp.\ to have non-negative Ricci curvature) if the scalar distribution  $\mathrm{Ric}(\mathcal{X},\mathcal{X})$ is non-negative for all smooth timelike vector fields $\mathcal{X}$. In the case of $g$ being smooth this condition coincides with the classical one, $\mathrm{Ric}(X,X)\geq 0$ for all timelike $X\in T_pM$ and all $p\in M$ by the fact that all such $X$ can be extended to smooth timelike vector fields on $M$. For the same reason the condition for $g\in C^{1,1}$ is equivalent to the condition $\mathrm{Ric}(\mathcal{X},\mathcal{X})\geq 0$ almost everywhere for all smooth timelike vector fields $\mathcal{X}$, used in the context of the $C^{1,1}$-singularity theorems. However, generalizing the null energy condition is more tricky due to the obstacles one encounters when extending null vectors. For a detailed discussion see \cite[Sec.\ 5]{Graf19}, and we define following her:

\begin{definition}[Distributional null energy condition]\label{def:nec}
A $C^1$-metric $g$ satisfies the \textit{distributional null energy condition}, if for any compact set $K\subseteq M$ and any $\delta>0$ there exists $\epsilon(\delta,K)$ such that $\text{Ric}(\mathcal{X},\mathcal{X})>-\delta$ (in the sense of distributions) for any local smooth vector field $\mathcal{X} \in \mathfrak{X}(U)$, $U\subseteq K$ with $\|\mathcal{X}\|_h=1$ and which is $\epsilon (\delta, K)$ close to a $C^1$ $g$-null vector field $\mathcal{N}$ on $U$, i.e.\  $||\mathcal{X}-\mathcal{N}||_h<\epsilon(\delta, K)$ on $U$.
\end{definition}

Again this condition is equivalent to the classical null energy condition if the metric is smooth. Moreover, in case $g\in C^{1,1}$ it is equivalent to the condition used in the $C^{1,1}$-setting i.e.\  $\text{Ric}(\mathcal{X},\mathcal{X})>0$ for all Lipschitz-continuous local null vector fields $\mathcal{X}$.
\medskip

One key technique in low regularity Lorentzian geometry is regularization. More specifically Chrusciel and Grant in \cite{CG12} have put forward a technique to regularize a continuous metric $g$ by smooth metrics $\check g_\varepsilon$ with narrower lightcones resp.\ by a net $\hat g_\varepsilon$ with wider lightcones than $g$. The basic operation  (denoted by $*$) is chartwise convolution with a standard molifier  $\rho_\epsilon$, which is globalized using cut-offs and a partition of unity, cf.\ \cite[Thm.\ 3.2.10]{GKOS01}. To manipulate the lightcones in the desired way one has to add a ``spacelike correction term''. The most recent version of this construction, which also quantifies the rate of convergence in terms of $\epsilon$ is \cite[Lem.\ 4.2]{Graf19}, which we recall here.


\begin{Lemma}\label{lem:2.2}
	Let $(M,g)$ be a spacetime with a $C^1$-Lorentzian metric. Then for any $\epsilon>0$, there exist smooth Lorentzian metrics $\check{g}_\epsilon$ and $\hat{g}_\epsilon$ with $\check{g}_\epsilon \prec g \prec \hat{g}_\epsilon$, both
converging to $g$ in $C^1_\text{loc}$. Additionally, on any compact set $K$ there is $c_k>0$ such that for all small 
$\epsilon$ 
\begin{equation}\label{eq:neweps}
 \| \check{g}_\epsilon - g * \rho_\epsilon\|_{\infty,K} \leq c_K \epsilon
 \quad\text{and}\quad
 \| \hat{g}_\epsilon - g * \rho_\epsilon\|_{\infty,K} \leq c_K \epsilon\,.
\end{equation}
\end{Lemma}

A main step in the proof of singularity theorems in low regularity is to show that the energy condition (Definition \ref{def:nec}, in our case) implies that the regularized metrics $\hat g_\epsilon$ and/or $\check g_\epsilon$ violate the classical energy conditions (the NEC, in our case) only by a small amount --- small enough, such that null geodesics still tend to focus. Technically this is done by a Friedrichs-type lemma, which in the present case is \cite[Lem.\ 4.5]{Graf19}, and draws essentially from \eqref{eq:neweps}. The corresponding result is then \cite[Lem.\ 5.5]{Graf19}:

\begin{Lemma}[Surrogate energy condition]\label{lem:sec} Let $M$ be a $C^1$-spacetime where the distributional null energy condition holds. Given any compact set $K\subseteq M$ and $c_2>c_1>0$, then for all $\delta>0$ there is $\epsilon_0>0$ such that $\forall \epsilon<\epsilon_0$
\begin{align}\nonumber
\mathrm{Ric}[\check{g}_\epsilon](X,X) > -\delta\quad 
  \forall X \in TM|_K  &\text{ with } \check{g}_\epsilon (X,X)=0\\
  &\text{ and } 0<c_1 \leq ||X||_h\leq c_2\,.
\end{align}
\end{Lemma}

Here $\check{g}_\epsilon$ is as in Lemma \ref{lem:2.2} and $\mathrm{Ric}[\check{g}_\epsilon]$ is its Ricci tensor. We will use this result in an essential way, when showing compactness of the horismos of a certain set, which is needed for the causal/analytic part of the proof of our Gannon-Lee theorem. A main difference to the arguments in \cite{Graf19} is that we do not assume global hyperbolicity, but are able to compensate for it by assuming non-branching of null maximizers.
Formally we define:

\begin{definition}\label{def:branch}
Let $M$ be a spacetime and let $\gamma:[0,1]\to M$ be a maximizing null curve. We say that $\gamma$ branches if there exists another maximizing null curve $\sigma:[0,1]\to M$ such that $\gamma(t)=\sigma(t)$ for all $0\leq t\leq a$ for some $0<a<1$ and $\gamma(t)\neq \sigma(t)$ for all $1\geq t>a $. The point $\gamma(a)$ is called branching point. 
If no maximizing null curve branches, we say that there is \textit{no null branching} in $M$. 
\end{definition}

In light of the fact that (even) for $C^1$-metrics causal maximizers are geodesics (to be proven in Theorem \ref{thm:maxgeod}, below), we see that if null branching were to occur in $\gamma(a)$,  there would be (at least) two different, maximizing solutions to the geodesic equations with initial values $\gamma(a)$ and $\gamma'(a)$. 

Generally speaking, in the low regularity Riemannian setting, branching of maximizers is associated with unbounded sectional curvature from below. More precisely, in length spaces with a lower curvature bound branching does not occur \cite[Lem.\ 2.4]{Shi93}. In smooth Lorentzian manifolds sectional curvature bounds, although more delicate to handle, are still characterized by triangle comparison \cite{AB08} and in the setting of Lorentzian length spaces \cite {KS18} synthetic curvature bounds extending sectional curvature bounds for smooth spacetimes have been introduced. In \cite[Section 4]{KS18} the authors show that a synthetic curvature bound from below prevents the branching of timelike maximizers. Unfortunately it is not clear at the moment how one could extend such a result to triangles with null sides.

However, in a merely $C^1$-spacetime the curvature is generically not locally bounded and so it seems natural that an additional condition as in Definition \ref{def:branch} is needed.
Indeed, this condition enters in an essential way into our arguments. It will be a topic of future investigations to relate null branching to curvature bounds. 
\medskip

\subsection{Results} 
Next we introduce the particular notions needed for the formulation of our results. In spirit they reflect Gannon's \cite{Gan75} assumptions on the spacetime, inspired
by asymptotic flatness, however, we stay close to the formulations of \cite[Sec.\ 2]{Sil10}.

\begin{definition}
An \emph{asymptotically regular hypersurface} is a spacelike, smooth, connected partial Cauchy surface $\Sigma$ which possesses an enclosing surface $S$,
i.e.\ a compact, connected submanfiold of codim.\ 2 in $M$ with the properties
\begin{enumerate}
 \item $S$ separates $\Sigma$ into two (open, sub-) manifolds $\Sigma_+, \Sigma_-$ such that $  
  \bar \Sigma_+ $ is non-compact and $\Sigma_-$ connected,
 \item the map $h_\#:\pi_1(S) \to \pi_1 (\bar \Sigma_+ )$, induced by the inclusion $h: S \to \bar \Sigma_+$ is surjective,
 \item $k_- >0 $ on $S$, i.e.\ $S$ is inner trapped.
\end{enumerate}

We further say that a surface $\Sigma$ admits a \emph{piercing}, if there exists a timelike vector field $X$ on $M$ such that every integral curve of $X$ meets $\Sigma$ exactly once.
\end{definition}

To elaborate on this definition, we first fix some notation.
Throughout let us denote the future directed timelike unit vector field perpendicular to $\Sigma$ near $S$ by $U$. Further as $S$ is a hypersurface in $\Sigma$ we denote by $N_\pm$ the unit normals to $S$ in $\Sigma$ such that $N_-$ points into $\Sigma_-$ and $N_+$ into $\Sigma_+$. We obtain future directed null normals to $S$ via $K_\pm := U|_S + N_\pm$. We will refer to $K_-(p)$ as the ingoing null vector. The convergence of a point in $S$ is defined via $k_\pm:=g(H_p,K_\pm(p))$, where $H_p$ is the mean-curvature vector field of $\Sigma$ at $p$. Observe that since $g\in C^1$, $H_p$ is still continuous and all corresponding formulae hold ``classically''.

Also note that any piercing of $\Sigma$ induces a continuous, open map $\rx:M \to \Sigma$, which maps any point $p$ in $M$ to the unique intersection of the integral curve of $X$ through $p$ with $\Sigma$. However, existence of a piercing is a strictly weaker condition than global hyperbolicity, see e.g.\ \cite[p.\ 4, 2nd paragraph]{Sil10}, although it implies that $M$ is homeomorphic to the product $\R\times\Sigma$ as pointed out in \cite[below 3.2]{MinSil20}. 
\medskip

We are now ready to state our main result, a Gannon-Lee theorem for non-globally hyperbolic $C^1$-spacetimes without null branching. We will discuss several of its special cases below.

\begin{Theorem}[$C^1$-Gannon-Lee theorem]\label{thm:glc1}
Let $(M,g)$ be a past reflecting, null geodesically complete $C^{1}$-spacetime without null branching and such that the distributional null energy condition holds. Let $\Sigma$ be an asymptotically regular hypersurface (with enclosing surface $S$) admitting a piercing. Further let one of the following possibilities hold
\begin{enumerate}[(i)]
	\item any covering spacetime of $(M,g)$ is past reflecting, or
	\item $S$ is simply connected and the universal covering spacetime of $(M,g)$ is past reflecting.
\end{enumerate}
Then the map $i_\# :\pi_1(S) \to \pi_1(\Sigma)$, induced by the inclusion $i: S \to \Sigma$, is surjective.
\end{Theorem}

Note that for $C^{1,1}$-metrics the geodesic equation is uniquely solvable and hence there can be no null branching. Moreover using that the distributional null energy condition reduces to the ``almost everwhere condition'' for $C^{1,1}$-metrics we immediately obtain the following $C^{1,1}$-Gannon-Lee theorem.

\begin{Corollary}[$C^{1,1}$-Gannon-Lee theorem]
Let $(M,g)$ be a past reflecting, null geodesically complete $C^{1,1}$-spacetime such that the null energy condition $\text{Ric}(\mathcal{X},\mathcal{X}) \geq 0 $ holds for all local Lipschitz-continuous null vector fields $\mathcal{X}$. Let $\Sigma$ be an asymptotically regular hypersurface (with enclosing surface $S$) admitting a piercing. Further let one of the assumptions (i) or (ii) of Theorem \ref{thm:glc1} hold.
Then the map $i_\# :\pi_1(S) \to \pi_1(\Sigma)$ is surjective.
\end{Corollary}

Going back to $C^1$ and assuming global hyperbolicity we clearly can skip past reflectivity and any assumption on the covering spacetime. Due to results in \cite{Graf19}, the assumption of no null branching is also obsolete. Moreover, in this case also some of the assumptions in \cite{Sil10} can be dropped, as they are implied by the existence of a Cauchy surface.

\begin{Corollary} [Globally hyperbolic $C^1$-Gannon-Lee theorem]\label{cor:ghc1}
Let $(M,g)$ be a globally hyperbolic, null geodesically complete $C^{1}$-spacetime such that the distributional null energy condition holds. Further let $\Sigma$ be an asymptotically regular hypersurface (with enclosing surface $S$), then the map $i_\# :\pi_1(S) \to \pi_1(\Sigma)$ is surjective.
\end{Corollary}

A simpler formulation is obtained assuming that $S$ is simply connected: the theorems then state that the entire spacetime is simply connected provided it is null complete. Originally the theorem of Gannon was given in contrapositive form, saying that if $S$ is (topologically) a sphere and $\Sigma$ is not simply connected, then $M$ has to be null incomplete.
\medskip

In proving our results we will follow the general layout of \cite{Sil10}. The proof consists of a causal and an analytic part as well as a purely topological part. At the heart of the causal part lies Proposition 4.1 of \cite{Sil10}, which we will prove for past reflecting $C^1$-spacetimes without null branching in Proposition \ref{prop:cp-c1}. It essentially states that the inside region $\Sigma_-$ of an asymptotically regular hypersurface $\Sigma$ is relatively compact. We will first establish the needed causality properties for $C^1$-spacetimes.

\section{Maximizers and Causality in $C^1$}\label{sec:c1}

In this section we establish properties of geodesics and results on causality in $C^1$-spacetimes. Building on recent results of \cite{GraLin18} and \cite{Graf19} we will establish that causal maximizers in $C^1$-spacetimes are geodesics. Note that this result was also independently discovered very recently in \cite{LLS20}. 

First note that by \cite[Thm.\ 1.1]{GraLin18} maximizers have a causal character, even in Lipschitz spacetimes.

We start by showing that also in a $C^1$-spacetime broken causal geodesics are not maximizing. As a prerequisite we use a variational argument similar to the one in \cite{ON83}, 10.45-46, which still holds true in our setting.

\begin{Lemma}\label{l3}
Let $c:[0,1]\to M$ be a causal pieceweise $C^2$-curve in a $C^1$-spacetime $M$ and let $X$ be a piecewise $C^1$-vector field along $c$. Then there is a piecewise $C^2$-variation $c_s$ of $c$ with variation vector field $X$. Moreover, if $g(X',c')<0$ along $c$ then for any variation $c_s$ of $c$ with variation vector field $X$ and small enough $s$, the longitudinal curve $c_s$ is timelike and longer than $c$.
\end{Lemma}

\begin{proof}
First we expand $X$ to a vector field $\tilde X$ in a neighbourhood of $c$ and set $c_s(t):= \text{Fl}^{\tilde X}_s(c(t))$, which is a variation of $c$ with variation vector field $X$.	
For the second part note that $g(c_0'(t),c_0'(t))\leq 0$ for all $t$ (except possible break points) as $c$ is causal and further 
\begin{equation*}
\frac{\partial}{\partial s}|_0 g(c_s'(t),c_s'(t)) = 2 g(\frac{\partial}{\partial s}|_0 \frac{\partial}{\partial t} c_s(t),c_0'(t)) = 2 g(\frac{\partial}{\partial t} X(t),c'(t)) <0
\end{equation*}
by assumption. So for small $s$ we have $g(c_s'(t), c_s'(t))< g(c'(t),c'(t))\leq 0$ (for almost all $t$) and hence $L(c_s)= \int_0^1 (-g(c_s'(t), c_s'(t)))^{\frac{1}{2}}\, dt > \int_0^1 (-g(c_0'(t), c_0'(t)))^{\frac{1}{2}}\, dt =L(c)$.
\end{proof}

\begin{Lemma}\label{l4}
In a $C^1$-spacetime no broken causal geodesic is maximizing.
\end{Lemma}

\begin{proof}
Let $\gamma:[0,2] \to M$ be a broken causal geodesic with a break point at $\gamma(1)$. Hence $v:=\lim_{t \uparrow 1} \gamma'(t)$ and $w:=\lim_{t \downarrow 1} \gamma'(t)$ are linearly independent. Also since both vectors are either future or past pointing, we have $\la v,w \ra <0$.

Let us show that the parameteriziation of $\gamma$ can be chosen such that:
\begin{equation}\label{eq:vw}
 \la v,v\ra - \la v, w\ra >0 \quad \text{as well as}\quad
 \la v,w \ra - \la w,w \ra <0.
                \end{equation}
If both $v$ and $w$ are null, this is clear by $ \la v,w \ra <0$. If both vectors are timelike we can w.l.o.g.\ assume them to be unit vectors. By the reverse Cauchy-Schwarz inequality we then have 
$| \la v,w \ra | > 1$ (recall that $v$ and $w$ are not colinear) and hence $\la v,w \ra < -1$. So $\la v,v \ra - \la v,w \ra   > 0$, and in the same way it follows that $\la v,w \ra - \la w,w\ra <0$.
		
		We now show that there exists a timelike curve from $\gamma(0)$ to $\gamma(2)$ longer than $\gamma$. First note that by continuity of the the Christoffel symbols we can solve the linear equations for parallel transport along $\gamma$ (uniquely) and the solution is a $C^1$-vector field.

		We set $y= v-w$ and let $Y_1$ and $Y_2$ be the vector fields along $\gamma|_{[0,1]}$ and $\gamma|_{[1,2]}$, obtained by parallel transport of $y$ along $\gamma|_{[0,1]}$ and $\gamma|_{[1,2]}$, respectively. Next we define a piecewise $C^1$-vector field $Y$ along $\gamma$ by setting $Y|_{[0,1]}=Y_1$ and $Y|_{[1,2]}=Y_2$.
		       
		Since $\gamma$ is a geodesic and $Y$ is parallel we have on $[0,1]$ using
		\eqref{eq:vw}
		\begin{equation}\label{eq:vw1}
		\la \gamma'(t), Y(t) \ra = \la v,v-w \ra = \la v,v\ra -\la v,w\ra >0 
		\end{equation}
		and similarly on $[1,2]$ we have $\la \gamma'(t),Y \ra <0$. 
		
		Now let $f:[0,2]\to [0,\infty)$ be a continuous, piecewise linear function such that $f(0)=0=f(2)$, $f'|_{[0,1)}=1$, and $f'|_{(1,2]}=-1$ and set $X(t):= f(t) Y(t)$. Then $X$ is a piecewise $C^1$-vector field along $\gamma$. Now consider the variation $\gamma_s$ of $\gamma$ with variation vector field $X$. By \eqref{eq:vw1} we have $\la \gamma',X' \ra <0$ and hence by Lemma \ref{l3} there exists some small $s$, such that $\gamma_s$ is longer than $\gamma$. By the choice of $f$ the endpoints agree and we have shown the statement.
	\end{proof}

\begin{Theorem}\label{thm:maxgeod}
Let $(M,g)$ be a $C^1$-spacetime, then any maximizing causal curve is a 
causal geodesic and hence a $C^2$-curve.
\end{Theorem}

\begin{proof}
Observe that being a geodesic is a local property and that any part of a maximizing curve is maximizing. Moreover, since any point in a $C^1$-spacetime has a globally hyperbolic neighbourhood\footnote{Actually there exists a smooth metric with wider lightcones, which has a neighbourhood base of globally hyperbolic sets, but these are also globally hyperbolic for $g$.}, we can assume $M$ to be globally hyperbolic. 

Let $\gamma:[0,1]\to M$ be a maximizer and set $p=\gamma(0)$, $q=\gamma(1)$.	
By \cite{Graf19}, Proposition 2.13, there exists a maximizing causal geodesic from $p$ to $q$ of the same causal character as $\gamma$. Also there
exist maximizing, causal geodesics $\sigma_1$ from $p$ to $\gamma(\frac{1}{2})$ and $\sigma_2$ from $\gamma(\frac{1}{2})$ to $q$. Note that since $\sigma_i$ are maximizing, we have $L(\sigma_1)=L(\gamma|_{[0,\frac{1}{2}]})$ and $L(\sigma_2)=L(\gamma|_{[\frac{1}{2},1]})$. This means $L (\sigma_1 \circ \sigma_2)=L(\gamma)$. So the curve $\gamma_1:=\sigma_1 \circ \sigma_2$ is maximizing and hence by Lemma \ref{l4} it is an unbroken geodesic.

This procedure can be iterated to obtain a sequence of maximizing causal geodesics $\gamma_n$ from $p$ to $q$, which meet $\gamma$ at all parameter values $\frac{k}{2^n}$, for $\mathbb{N}\ni k \leq 2^n$. Observe that $\gamma_n$ converge to $\gamma$ uniformly: First, for any $\epsilon>0$  we cover $\gamma$ by finitely many open, causally convex, sets $V^\epsilon_{p_i}$ around $p_i\in\gamma$ of $h$-diameter at most $\epsilon$. The union $V^\epsilon=\sup_i V^\epsilon_{p_i}$ is a neighbourhood of $\gamma$, and there exist dyadic numbers $t_m$, $m=0,\ldots k$, $t_0=0$ and $t_k=1$, such that $\gamma(t_m)$ and $\gamma(t_{m+1})$ lie in a single $V^\epsilon_{p_i}$ for some $i$. Moreover, there exists some $N(\epsilon)$, such that for all $n \geq N$ all curves $\gamma_n$ meet every $\gamma(t_m)$ and hence the segments of $\gamma_n$ from $t_m$ to $t_{m+1}$ are contained in $V^\epsilon_{p_i}$. So we conclude that $d_h(\gamma(t),\gamma_n(t))\leq \epsilon$, so $\gamma_n \to \gamma$ uniformly.

We can now parameterize $\gamma_n$, such that $\|\gamma_n'(0)\|_h=1$ and hence pass to a subsequence (again denoted by $\gamma_n$) such that $\gamma_n'(0) \to v$.
By \cite[Chap.\ II Thm.\ 3.2,]{Hart02}, there exists a subsequence of $\gamma_n$ which converges uniformly on compact sets to a geodesic $\sigma$ with initial values $\sigma(0)=p$ and $\sigma'(0)=v$. But as $\gamma_n$ converges to $\gamma$, so must any subsequence and hence $\gamma=\sigma$ on the entire domain of $\gamma$. Finally $\sigma$ reaches $q$ since otherwise, $\sigma$ would agree with $\gamma$ on its entire maximal domain of definition, would be future inextendible and contained in the compact set $\gamma([0,1])$, a contradiction to non-imprisonment which holds in any globally hyperbolic $C^1$-spacetime. 
\end{proof}

Assuming non-branching we are able to prove results on limits of maximizers needed in the following.

\begin{Proposition}\label{Prop:nonbranch}
Let $M$ be a $C^1$-spacetime without null branching. 
If two causal geodesic segments contained in an achronal set intersect, they are segments of the same geodesic or they intersect at the endpoints.
\end{Proposition}

\begin{proof}
	Suppose such a segment intersetcs a second one in the interior of its domain. Then either their tangents at the meeting point are not proportional and so by Lemma \ref{l4} their concatination, which is a broken causal geodesic, stops maximising, contradicting the fact that both segments are contained in an achronal set. Or otherwise their tangents at the meeting point are proportional and hence null branching would occur, again a contradiction.
\end{proof}

Using Theorem \ref{thm:maxgeod} we may also give a slightly different formulation: 
If two different initially maximizing null curves starting at the same point meet again, they stop maximizing.
\medskip

The final result of this section will be essential in the proof of the main theorem and it is the main point at which we use the assumption that null branching doesn't occur. Again, $\check g_\epsilon$ is as in Lemma \ref{lem:2.2}.

\begin{Corollary}\label{Cor:epslimitmax}
Let $(M,g)$ be a $C^1$-spacetime without null branching.
Let $S$ be an enclosing surface in a partial Cauchy surface $\Sigma$ and let $\gamma:[0,1]\to E_g^+(S)$ be a $g$-null curve. Then for any $1>\delta>0$, $\gamma|_{[0,1-\delta]}$ is a limit of $\check{g}_{\epsilon_n}$-null curves contained in $E_{\check{g}_{\epsilon_n}}^+(S)$ for an appropriate subsequence of $\check{g}_{\epsilon_n}$.
\end{Corollary}

\begin{proof}
Let $\gamma$ be a future directed $g$-null $S$-maximizer starting at $p=\gamma(0)\in S$, so $\gamma\subseteq E^+(p)$. For any $1>\delta>0$ there exist points $ q_n^\delta \in \pt I^+_{n}(S):=\pt I^+_{\check g_{\epsilon_n}}(S)$ converging to $\gamma(1-\delta)=:q^\delta$: To see this let $U_k$ be a sequence of connected nested neighbourhoods of $q^\delta$ with $\bigcap_k U_k=\{q^\delta \}$. Choose some $q^e_k \in U_k \bs \overline{J^+(S)} $
and $q^i_k \in U_k \cap I^+(S)$. For large $n$ we can achieve $q^i_k \in I^+_{{n}}(S)$ and since also $q^e_k \in U_k\bs \overline{J^+_{{n}}(S)}$ there exists a curve from $q^i_k$ to $q^e_k$ which starts in $I^+_{{n}}(S)$ and leaves it and hence must meet $\pt I^+_{n}(S)$ in a point which we call $q_n^\delta$.

Hence there are future directed $\check{g}_{\epsilon_n}$-null maximizing geodesics $\gamma_n^\delta$ ending at $q_n^\delta$ and contained in $\pt I^+_{{n}}(S)$. Note that the $\gamma_n^\delta$ either meet $S$ or are past inextendible.

Now by \cite[Chap.\ II, Thm.\ 3.2]{Hart02} there exists a subsequence of $\gamma_n^\delta$, denoted again by $\gamma_n^\delta$, converging to a maximizing $g$-null geodesic $\sigma^\delta$ ending at $q^\delta$ which is entirely contained in $\pt I^+(S)$. By Thm.\ \ref{thm:maxgeod} $\gamma$ is a geodesic and continues to be maximizing after $q^\delta$. Also, by construction $\sigma^\delta$ is non-trivial, so $\gamma$ and $\sigma^\delta$ coincide on some $\gamma|_{[a,1-\delta]}$ (with $0\leq a<1-\delta$) by Proposition \ref{Prop:nonbranch}.

There are two possibilities: Either there exists a subsequence of $\check{g}_{\epsilon_n}$ such that any $\gamma_n^\delta$ meets $S$, in which case, by passing to this subsequence, also $\sigma^\delta$ meets $S$ and hence $\sigma^\delta\supseteq \gamma|_{[0,1-\delta]}$ and we obtain the desired property.

The other possibility is that there is no subsequence, such that all $\gamma_n^\delta$ meet $S$. We show that this is impossible. If this were the case we could choose a subsequence $\gamma_n^\delta$ of past-inextendible curves. Then also $\sigma^\delta$ is past-inextendible and hence must leave $\gamma([0,1-\delta])$. Thus again $\sigma^\delta\supseteq \gamma|_{[0,1-\delta]}$ and it remains to show that $\gamma_n^\delta \subseteq E^+_n(S)$ for large $n$. 

To this end let $U:=D(\Sigma)$, which is a globally hyperbolic, open, and causally convex  neighbourhood of $p=\gamma(0)\in S$ for the metric $g$ \cite[Cor.\ 3.36, Prop.\ 3.43]{Min19}\footnote{The proofs carry over verbatim to $C^1$-metrics.} and hence also for $g_{\varepsilon_n}$\footnote{This follows immediately since we approximate from  the inside. However, global hyperbolicity is stable in the interval topology even for continuous metrics \cite{Sam16} and $C^1$-convergence is stronger.}. Then there are points $p_n \in \gamma_n^\delta$ with $p_n\to p$ and so for large $n$, $p_n \in \pt I^+_n(S)\cap U= \pt I^+_n(S,U)=E^+_n(S,U)$. Hence there is a $g_{\varepsilon_n}$-null geodesic from $S$ to $p_n$ contained in $E^+_n(S,U)$ which by Proposition \ref{Prop:nonbranch} coincides with (a part of) $\gamma_n^\delta$. Hence $\gamma_n^\delta$ must meet $S$, a contradiction.
\end{proof}

Observe that Corollary \ref{Cor:epslimitmax} clearly holds true for $C^{1,1}$-spacetimes. For globally hyperbolic $C^1$-spacetimes (where, in principle, null branching can occur) a similar result was established in \cite[Prop.\ 2.16]{Graf19}, with slightly different assumptions on $S$ and a weaker conclusion, which only guarantees that for any $p\in E^+_g(S)$, there is a geodesic segment which is an appropriate limit of approximating $\check{g}_\epsilon$-maximizers.

\section{Proof of the main result}\label{sec:proof} We will split the proof of Theorem \ref{thm:glc1} into two parts. The first, analytic part will be concerned with showing that the set $\Sigma_-$ is relatively compact. Here we will generalize (the proof of) \cite[Prop.\ 4.1]{Sil10} by proving new focusing statements for null geodesics using the results from section \ref{sec:c1}. 
The second, topological part uses the results from causality theory detailed in section \ref{sec:c1} to generalize (the proof of) \cite[Thm.\ 2.1]{Sil10}. 

To be self-contained we will briefly sketch also those parts of the original (smooth) proofs which do not need major revision.

\subsection{Analytic aspects}

The main analytical ingredient of our proof is a generalization of \cite[Prop.\ 4.1]{Sil10} to $C^1$-spacetimes, which we will give below in Proposition \ref{prop:cp-c1}. 

In order to do so we need a focusing result for null geodesics in smooth spacetimes which violate the null energy condition by a small margin $\delta$, as in Lemma \ref{lem:sec}. To this end we apply a result of \cite{FK19}\footnote{Observe the opposite signature convention used there.}, which itself is a generalization of \cite[Prop.\ 10.43]{ON83}.

\begin{Proposition} (\cite[Prop.\ 2.7]{FK19}) Let $S$ be a spacelike submanifold of co-dimension $2$ 
in a smooth spacetime and let $\gamma$ be a null geodesic joining $p\in S$ to $q\in J^+(S)$. If there exists a smooth function $f$ on $\gamma$ which is nonvanishing at $p$ but vanishes at $q$ and so that
\begin{equation}\label{eqn:nullgen}
 \int_\gamma \left( (n-2)(f')^2 -f^2 \, \mathrm{Ric}(\gamma',\gamma') \right) \  \leq\ (n-2)\ \langle f^2 \,
 \gamma', H)\rangle\,|_{p} \,,
\end{equation}
then there is a focal point to $S$ along $\gamma$.
\end{Proposition}
 
\begin{Lemma}\label{lem:deltafocusing}
Let $S$ be a $C^2$-spacelike submanifold of codimension 2 in a smooth spacetime. Let $\gamma$ be a geodesic starting at some $p\in S$ such that $\nu:=\gamma'(0)$  is a future pointing null normal to $S$. Let the convergence
\begin{equation}\label{eq:fk} 
 c:=k_S(\nu):= \langle H_{\gamma(0)},\nu \rangle >0
\end{equation}
and choose some $b>\frac{1}{c}$ and $0<\delta (b,c)=:\delta \leq \frac{3}{b^2} (n-2)(bc-1)$. Now, if $\text{Ric}(\gamma',\gamma') \geq - \delta$ along $\gamma$, then $\gamma|_{[0,b]}$ cannot be maximizing to $S$, provided it exists that long.\footnote{One can specify the choice of $b$ further to allow for bigger violations of $\delta$, but we do not need this here.}
\end{Lemma}

\begin{proof}
We set $f(t):= 1-\frac{t}{b}$ and check condition \eqref{eqn:nullgen}. For our choice of $\delta$ we obtain
\begin{align}\nonumber
 \int_{0}^{b}&(n-2)\,\frac{1}{b^2}\, dt - \int_{0}^{b} \left(1-2\, \frac{t}{b}+\frac{t^2}{b^2}\right)\, \text{Ric}(\gamma'(t),\gamma'(t)) \, dt  \\
 &\leq 
\frac{n-2}{b}+ \delta\, \frac{b}{3} \leq (n-2)\, c =(n-2)\, k_S(\nu)
=(n-2)\, \langle f^2 \gamma', H)\rangle\,|_{p} \,,
\end{align}
and hence $\gamma|_{[0,b]}$ cannot be maximizing.
\end{proof}

Recall that for $C^{1}$-metrics the mean curvature and the convergence of $S$ are still continuous.
The core of the following proof is largely from \cite{MinSil20}.

\begin{Proposition}\label{prop:cp-c1}
Let $(M,g)$ be an $n$-dimensional (with $n\geq 3$), past reflecting, null geodesically complete \emph{$C^{1}$}-spacetime without null branching which satisfies the distributional null energy condition and admits an asymptotically regular hypersurface $\Sigma$ with a piercing. Then for any enclosing surface $S\subseteq\Sigma$ the closure of its inside, $\overline\Sigma_-=S\cup\Sigma_-$ is compact.
\end{Proposition}
    
We consider the closed set $T:=\partial I^+(\Sigma_+)\setminus\Sigma_+$. Also, by Theorem \ref{thm:maxgeod} the set $E^+(S)$ really consists of null \emph{geodesics} emanating from $S$ 
and perpendicular to $S$, which follows as in \cite[10.45, 10.50]{ON83} by using Lemma \ref{l3}.
So we define $\mathcal{H}^+$ as the subset of all points $p\subseteq E^+(S)$ on future directed null geodesics $\gamma:[0,1]\to M$ with $\gamma(0)\in S$, $\gamma'(0)=K_-(\gamma(0))$. Note that $S\subseteq\mathcal{H}^+$. Further, no point on $\mathcal{H}^+$ can lie on a null geodesic from $S$ in direction of $K_+$, see e.g.\ \cite[Lemma 1.1]{Gan75}\footnote{The proof there is given for smooth spacetimes and one assumes $S$ to be simply connected, however the method of proof also works in our case. }

The proof consists in successively establishing the following three claims: 
 \medskip
 \begin{enumerate}
  \item [(1)] $\mathcal{H}^+$ is relatively compact, 
  \qquad (2) $T$ is compact,
  \qquad (3) $\rho_X(T)=\overline{\Sigma}_-$.
 \end{enumerate}
\medskip

\noindent
Steps (1) and (2) combine the causality part of the proof with the analytical arguments which we have to provide in $C^1$-regularity. Some arguments in steps (2) and (3) do not require changes from the original proofs put forward in \cite{MinSil20} resp.\ \cite{Sil10} but will be included as a sketch for the sake of completeness.

\begin{proof}
	
(1) Any point $p$ on $\mathcal{H}^+$ lies on a null geodesic emanating from $S$ and its initial tangent vector is inward pointing, i.e.\ proportional to $K_-(q)$ for some $q \in S$. 
By continuity $k_-$ possesses a minimum $c:= \min_{p\in S} k_-(p)= \min_{p\in S} \langle H_p,K_-(p) \rangle$ on $S$. 
Also, the set $K:= \{ (p,\lambda K_-(p)) \in TS ^\perp \, |\, 0 \leq \lambda \leq \frac{2}{c} \} \subseteq TM$ is compact and, by \cite[Prop.\ 2.11]{Graf19} (or rather a simplified version without $\epsilon$) the set 
\begin{equation}
F:= \bigcup_{\dot \gamma \text{ with } \dot \gamma(0) \in K} \text{im}(\dot \gamma|_{[0,1]}) 
\end{equation}
is relatively compact (Here $\dot \gamma$ denotes the trajectory of a geodesic in $TM$ with the specified initial conditions).

We will show that $\mathcal{H}^+ \subseteq \pi( F)$, where $\pi:TM \to M$ is the projection.
Assume the contrary and let $p \in \mathcal{H}^+ \bs \pi(F)$. Let $\gamma:[0,1]\to M$ be a null geodesic from $S$ to $p$ maximizing the distance to $S$. As $\gamma \subseteq \mathcal{H}^+$ we know that $\gamma'(0)= \mu K_-(\gamma(0))$ for some $\mu>\frac{2}{c}$, since $\gamma'(0) \not \in K$. This means that $k_S(\gamma'(0))= \langle H(\gamma(0)), \mu K_-(\gamma(0)) \rangle \geq \mu c > 2$.

Let $\check g_{\epsilon_k}$ be as in Lemma \ref{lem:2.2}. By Corollary \ref{Cor:epslimitmax} $\gamma|_{[0,1-\delta]}$ for some arbitrarily small $\delta$ is the $C^1$-limit of $\check g_{\epsilon_k}$-null geodesics $\gamma_{\epsilon_k}:[0,b_k]\to M$ with $b_k\to 1-\delta$ contained in $E^+_{\check g_{ \epsilon_k}}(S)$.
Further we can assume that all $\gamma_{\epsilon_k}$ are contained in a compact neighbourhood $\tilde K$ of $\gamma$ and that $c_1 < \| \gamma'_{\epsilon_k}\|_h < c_2$ for some $c_i>0$. Additionally for $k$ large enough we have $k_S^{\epsilon_k}(\gamma'_{\epsilon_k}(0)):=c_k>2$ and by Lemma \ref{lem:sec} we can also achieve $\text{Ric}[g_{\epsilon_k}](\gamma'_{\epsilon_k},\gamma'_{\epsilon_k}) \geq - 3(n-2)$. In order to apply Lemma \ref{lem:deltafocusing} for $b=1$ we set $\delta_k = \frac{3}{b^2}(n-2)(b \, c_k-1)=3(n-2)(c_k-1)$. Since $c_k>2$ for all large $k$, we have $-\delta_k<- 3 (n-2)$ and hence $\text{Ric}[g_{\epsilon_k}](\gamma'_{\epsilon_k},\gamma'_{\epsilon_k}) \geq -3(n-2) > -\delta_k$. So by Lemma \ref{lem:deltafocusing}, $\gamma_{\epsilon_k}$ cannot be maximizing up to $\frac{1}{c_k}< \frac{1}{2} <1-\delta<b$ but by construction it is maximizing up to $b_k\to 1-\delta$, a contradiction. 
\medskip 

(2) We prove the inclusion $T \subseteq \mathcal{H}^+$ which gives that $T$ is compact. Assume by contradiction that there were $q \in T\bs \mathcal{H}^+$. By \cite[Proof of 3.5]{MinSil20} we can find $q_n \in I^+(S)$ such that $q_n \to q$ and future directed, fututre inextendible timelike curves $\sigma_n:[0,\infty)\to M$ parametrized by $h$-arclength starting at $S$ with $\sigma_n(t_n)=q_n$.  
Further by the limit curve theorem, which is valid in $C^1$ spacetimes, see \cite[Thm.\ 14]{Min18} one obtains a future directed, future inextendible causal curve $\sigma$ starting at $S$. If $q$ were to lie on $\sigma$, it had to be a maximizing null geodesic perpendicular to $S$\footnote{by the same argument as above using \cite[10.45, 10.50]{ON83} and Lemma \ref{l3}}. 
It can however neither start inward going (in direction $K_-$) as then $q\in \mathcal{H}^+$ nor
outward going (in direction $K_+$) as then $q \in I^+(\Sigma_+)$.

Hence $q\not\in\sigma$ and so $t_n\to \infty$. Since by (1) there is no inward pointing $S$-null ray, there is $b\in(0,\infty)$ with $\sigma(b)\in I^+(S)$. But then, again by the limit curve theorem, $q \in \overline{I^+(\sigma(b))}$. By past-reflectivity one obtains $\sigma(b)\in \overline{I^-(q)}\cap  I^+(S)$, implying $q\in I^+(S)$ and hence $q \in I^+(\Sigma_+)$, contradicting $q\in T$.
\medskip

(3) First $\rx(T) \subseteq \overline \Sigma_-$ (since otherwise $T \cap I^+(\Sigma_+)\not=\emptyset$) and $\overline \Sigma_- \bs \rx(T) \subseteq \Sigma_-$ (since $S = \rx(S) \subseteq \rx(T)$). Now assuming indirectly that $\overline\Sigma_-\setminus\rx(T)\not=\emptyset$, there is $p\in \pt_\Sigma\rx(T)\cap\Sigma_-$ and we will reach a contradiction by showing that $p \in \text{int}\rx(T)$. 

By compactness of $T$ there is $q\in T$ with $\rx(q)=p$, and $q\not\in S$ (otherwise $p=q\in S\cap\Sigma_-=\emptyset$). So $q\in T\setminus S=\pt I^+(\Sigma_+)\setminus\overline\Sigma_+$ which is a topological hypersurface ($S$ being the edge of the achronal set $T$). So there is $V_0$, an $M$-neighbourhood of $q$ with $V_0 \cap T$ open in $T\bs S$, $\rx(V_0) \subseteq \Sigma_-$, and $V_0 \cap \Sigma= \emptyset$. Next denote by $\Psi$ the local flow of $X$ and choose $\epsilon >0$ and $U_0$, an open $M$-neighbourhood of $q$, so small that $\Psi(U_0 \times (-\epsilon, \epsilon))\subseteq V_0$. Further set $\Psi_0 := \Psi_{|(U_0\cap T)\times (-\epsilon,\epsilon)}$, and $W:= \text{Im}\Psi_0$. By achronality of $T$ and invariance of domain $W$ is open and $\Psi_0$ is a homeomorphism. But then $p\in\rx(U_0\times(-\epsilon,\epsilon))=\rx(W)$ and the latter set is open by openness of $\rx$ (\cite[14.31]{ON83}) and so $p\in \text{int}\rx(T)$.
\end{proof}

\subsection{Topological aspects}

Finally we invoke Proposition \ref{prop:cp-c1} to prove our main result.
Here we will be brief on the topological aspects laid out already in the proof of \cite[Thm.\ 2.1]{Sil10}.

\begin{proof}[Proof of theorem \ref{thm:glc1}] 
	
Let $\Phi:\tilde{M} \to M$ be a connected (smooth) covering with $\Phi_\# (\pi_1 (\tilde M)) = j_\#(\pi_1 (S))$, where $j$ is the inclusion of $S$ in $M$.  

Note that w.l.o.g.\ one can assume the vector field of the piercing to be complete and by properties of its flow map easily show that $M \cong \R \times \Sigma$. Hence the inclusion $m$ of $\Sigma$ in $M$ induces an isomorphism $m_\# : \pi_1(\Sigma) \to \pi_1(M)$. In particular, $\Phi_\Sigma := \Phi|_{\tilde{\Sigma}}: \Phi^{-1}(\Sigma):=\tilde{\Sigma} \to \Sigma$ is a Riemannian covering with $\tilde{\Sigma}$ connected. 

In Lemma \ref{lem:final} below we will establish that $\Phi_\Sigma$ is trivial. Accepting this for the moment, we will show the theorem, i.e.\  for every $y \in \pi_1(\Sigma)$ there is $x \in \pi_1(S)$ with $i_\#(x) =y$. From the diagrams\\
\begin{minipage}{0.4 \textwidth}
 \centering
 \begin{tikzcd}
  & M\\
  S \arrow[r,hook, "i"] \arrow[ru,hook,"j"] &\Sigma \arrow[u,hook, "m"] & \tilde{M} \arrow[lu,"\Phi"]  \\
  & \tilde{\Sigma}\arrow[u,leftrightarrow, "\Phi_\Sigma"]  \arrow[ru,"\tilde{m}"] 	
\end{tikzcd}
\end{minipage}
\begin{minipage}{0.55 \textwidth}
\centering
\begin{tikzcd}
 & & j_\#(\pi_1(S)) \arrow[d,hook]\\
 \pi_1(S) \arrow[r, "i_\#"] \arrow[urr, bend left, "j_\#"]  & \pi_1(\Sigma) \arrow[r, leftrightarrow,"m_\#"]& \pi_1(M)  \\
 & \pi_1(\tilde{\Sigma}) \arrow[r] \arrow[u]  & \pi_1(\tilde{M}) \arrow[u] 
 \arrow[uu,bend right=70,"\Phi_\#"]
\end{tikzcd}
\end{minipage}\\
we see that 
\begin{equation}
 m_\#(y) = (\Phi \circ \tilde{m} \circ \Phi_\Sigma^{-1})_\#(y) = \Phi_\# (\tilde{m} \circ \Phi_\Sigma^{-1})_\# (y) \in \Phi_\#(\pi_1(\tilde{M})) = j_\#(\pi_1(S)).
\end{equation}
Since $j=m \circ i$ we have $m_\#(y) = j_\#(x) = m_\#(i_\#(x))$, and since $m_\#$ is an isomorphism and we are done.
\end{proof}
    
\begin{Lemma}\label{lem:final} 
$\Phi_\Sigma := \Phi|_{\tilde{\Sigma}}: \tilde{\Sigma}:=\Phi^{-1}(\Sigma) \to \Sigma$ is a trivial covering.
\end{Lemma}

\begin{proof}
First, by definition there is a local deformation $F : S \times (-1,1) \to \Sigma$ of $S$. Further set $U_F^- := F(S \times (0,1))$ and $V:= U_F^- \cup S \cup \Sigma_+$.
Then $\bar \Sigma_+$ is a deformation retract of $V$ and hence $\pi_1(\bar \Sigma_+)\cong \pi_1(V)$.
\medskip

Next we establish that on every connected component $\tilde{V}$ of $\Phi_\Sigma^{-1}(V)$ the map $\Phi_V:=\Phi|_{\tilde{V}} : \tilde{V} \to V$ is a diffeomorphism.

We only have to show injectivity since $\Phi_V$ is a local diffeomorphism.
Take $\tilde{p}, \tilde{q} \in \tilde{V}$ such that $\Phi_V(\tilde{p})= \Phi_V(\tilde{q})=:p \in V$.
Let $\tilde{\alpha} :[0,1] \to \tilde{V}$ be a path connecting these two points, then $\alpha := \Phi \circ \tilde{\alpha}$ is a loop in $V$, homotopic to a loop in $S$ since $\pi_1(V) \cong \pi_1(\bar \Sigma_+) \cong \pi_1(S)$. Further since $\Phi_\#(\pi_1(\tilde{M}))= j_\#(\pi_1(S))$, there is a loop $\tilde{\beta}$ in $\tilde{M}$, fixed endpoint-homotopic to $\tilde{\alpha}$ and so we must have $\tilde{p}=\tilde{q}$.
\medskip

In order to show that $\Phi_\Sigma$ is trivial we assume the converse, and, in particular that $\Phi_\Sigma^{-1}(S)$ has more than one component. Since $S \subseteq V$ each of these components is diffemorphic to $S$, and they separate $\tilde \Sigma$. 

Let $\tilde{S}_1, \tilde{S}_2$ be two such different components and let 
$\tilde{V}_1, \tilde{V}_2$ be the respective copies of $V$ containing 
them. 
Since $\bar \Sigma_+ \subseteq V$, each $\tilde{V}_i$ ($i=1,2$) contains a diffeomorphic copy of $\bar \Sigma_+$ called $\tilde{C}_i$, which are closed and non-compact.
Further since $\tilde{\Sigma}$ is connected and separated by $\tilde{S}_1$, the set $\tilde{C}_2$ 
is contained in $\tilde{S}_1 \cup \tilde{\Sigma}_-^{(1)}:= \tilde{S}_1 \cup (\tilde{\Sigma} \bs \tilde{C}_1)$, since otherwise 
$\tilde{V}_1 \cap \tilde{V}_2 \neq \emptyset$.

So $\tilde{C}_2 \subseteq \tilde{S}_1 \cup \tilde{\Sigma}_-^{(1)}$. Being local properties, both null geodesic completeness and the distributional null convergence condition lift
to $\tilde M$ and by assumption $\tilde M$ is past-reflecting. 
Further as null branching is a local property as well, it also cannot occur in $\tilde M$. Thus the assumptions of Proposition \ref{prop:cp-c1} are fulfilled for $\tilde{M}$, $\tilde{\Sigma}$, $\tilde{C}_1$ and $\tilde{\Sigma}\bs \tilde{C}_1$ implying that $\tilde{S}_1 \cup (\tilde{\Sigma} \bs \tilde{C}_1)$ is compact. 
However this set contains the non-compact, closed set $\tilde{C}_2$, a contradiction and we are done.
\end{proof}

Finally we \emph{sketch the proof of Corollary \ref{cor:ghc1}}, i.e.\  the globally hyperbolic $C^{1}$-Gannon-Lee theorem. We can proceed analogously to the one of Theorem \ref{thm:glc1}, where most points will even be significantly easier. In fact the topological part remains the same and the only aspect one 
has to pay attention to is proving that ${\mathcal H}^+$ is relatively compact, i.e. step (1) in Proposition \ref{prop:cp-c1}: We have to prove that any maximizing $g$-null
geodesic is a limit of maximizing $\check g_\epsilon$-null geodesics. Corollary \ref{Cor:epslimitmax} does not apply, but we can replace it by the ``limiting-result'' \cite[Prop.\ 2.13]{Graf19}.
Note that compared to Corollary \ref{Cor:epslimitmax} the result in \cite[Prop.\ 2.13]{Graf19} only shows that there is one $g$-geodesic which can be approximated by $\check{g}_\epsilon$ maximizers, but this is sufficient for the proof to work out.
\medskip

\section*{Acknowledgement} We are greatful to Michael Kunzinger for sharing his experience and to Melanie Graf and Clemens Sämann for helpful discussions. We would especially like to thank Ettore Minguzzi for pointing out a mistake in an earlier version of this note, and to the anonymous referees for their valuable comments. This work was supported by FWF-grants  P28770, P33594 and the Uni:Docs program of the University of Vienna.



\begin{thebibliography}{99}

\bibitem[AB08]{AB08}
Stephanie~B. Alexander and Richard~L. Bishop.
\newblock Lorentz and semi-{R}iemannian spaces with {A}lexandrov curvature
  bounds.
\newblock {\em Comm. Anal. Geom.}, 16(2):251--282, 2008.

\bibitem[BS18]{BS18}
Patrick Bernard and Stefan Suhr.
\newblock Lyapounov functions of closed cone fields: from {C}onley theory to
  time functions.
\newblock {\em Comm. Math. Phys.}, 359(2):467--498, 2018.

\bibitem[CeS10]{Sil10}
Ivan~P. Costa~e Silva.
\newblock On the {G}annon-{L}ee singularity theorem in higher dimensions.
\newblock {\em Classical Quantum Gravity}, 27(15):155016, 13, 2010.

\bibitem[CeSM20]{MinSil20}
Ivan~P. Costa~e Silva and Ettore Minguzzi.
\newblock A note on causality conditions on covering spacetimes.
\newblock {\em Classical Quantum Gravity}, 37(22):227001, 12, 2020.

\bibitem[CG12]{CG12}
Piotr~T. Chru\'sciel and James D.~E. Grant.
\newblock On {L}orentzian causality with continuous metrics.
\newblock {\em Classical Quantum Gravity}, 29(14):145001, 32, 2012.

\bibitem[CK18]{CK18}
Piotr~T. Chru{\'{s}}ciel and Paul Klinger.
\newblock The annoying null boundaries.
\newblock {\em Journal of Physics: Conference Series}, 968:012003, feb 2018.

\bibitem[ER08]{Emp08}
Roberto Emparan and Harvey~S. Reall.
\newblock {Black Holes in Higher Dimensions}.
\newblock {\em Living Rev. Rel.}, 11:6, 2008.

\bibitem[FK20]{FK19}
Christopher~J. Fewster and Eleni-Alexandra Kontou.
\newblock A new derivation of singularity theorems with weakened energy
  hypotheses.
\newblock {\em Classical Quantum Gravity}, 37(6):065010, 31, 2020.

\bibitem[FSW95]{FSW95}
John~L. Friedman, Kristin Schleich, and Donald~M. Witt.
\newblock Comment on: ``{T}opological censorship'' [{P}hys. {R}ev. {L}ett. {\bf
  71} (1993), no. 10, 1486--1489; {MR}1234452 (94e:83071)].
\newblock {\em Phys. Rev. Lett.}, 75(9):1872, 1995.

\bibitem[Gal83]{Gal83}
Gregory~J. Galloway.
\newblock Minimal surfaces, spatial topology and singularities in space-time.
\newblock {\em J. Phys. A}, 16(7):1435--1439, 1983.

\bibitem[Gan75]{Gan75}
Dennis Gannon.
\newblock Singularities in nonsimply connected space-times.
\newblock {\em J. Mathematical Phys.}, 16(12):2364--2367, 1975.

\bibitem[Gan76]{Gan76}
Dennis Gannon.
\newblock On the topology of spacelike hypersurfaces, singularities, and black
  holes.
\newblock {\em General Relativity and Gravitation}, 7(2):219--232, 1976.

\bibitem[GGKS18]{GGKS18}
Melanie Graf, James D.~E. Grant, Michael Kunzinger, and Roland Steinbauer.
\newblock The {H}awking-{P}enrose singularity theorem for
  {$C^{1,1}$}-{L}orentzian metrics.
\newblock {\em Comm. Math. Phys.}, 360(3):1009--1042, 2018.

\bibitem[GKOS01]{GKOS01}
Michael Grosser, Michael Kunzinger, Michael Oberguggenberger, and Roland
  Steinbauer.
\newblock {\em Geometric theory of generalized functions with applications to
  general relativity}, volume 537 of {\em Mathematics and its Applications}.
\newblock Kluwer Academic Publishers, Dordrecht, 2001.

\bibitem[GKS19]{GKSS19}
James D.~E. Grant, Michael Kunzinger, and Clemens S\"{a}mann.
\newblock Inextendibility of spacetimes and {L}orentzian length spaces.
\newblock {\em Ann. Global Anal. Geom.}, 55(1):133--147, 2019.

\bibitem[GKSS20]{KGSS20}
James D.~E. Grant, Michael Kunzinger, Clemens S\"{a}mann, and Roland
  Steinbauer.
\newblock The future is not always open.
\newblock {\em Lett. Math. Phys.}, 110(1):83--103, 2020.

\bibitem[GL18a]{GL18}
Gregory~J. Galloway and Eric Ling.
\newblock Topology and singularities in cosmological spacetimes obeying the
  null energy condition.
\newblock {\em Comm. Math. Phys.}, 360(2):611--617, 2018.

\bibitem[GL18b]{GraLin18}
Melanie Graf and Eric Ling.
\newblock Maximizers in {L}ipschitz spacetimes are either timelike or null.
\newblock {\em Classical Quantum Gravity}, 35(8):087001, 6, 2018.

\bibitem[GLS18]{GSL18}
Gregory~J. Galloway, Eric Ling, and Jan Sbierski.
\newblock Timelike completeness as an obstruction to {$C^0$}-extensions.
\newblock {\em Comm. Math. Phys.}, 359(3):937--949, 2018.

\bibitem[Gra20]{Graf19}
Melanie Graf.
\newblock Singularity theorems for {$C^1$}-{L}orentzian metrics.
\newblock {\em Comm. Math. Phys.}, 378(2):1417--1450, 2020.

\bibitem[GS06]{GalSch06}
Gregory~J. Galloway and Richard Schoen.
\newblock A generalization of {H}awking's black hole topology theorem to higher
  dimensions.
\newblock {\em Comm. Math. Phys.}, 266(2):571--576, 2006.

\bibitem[GT87]{GT87}
Robert Geroch and Jennie Traschen.
\newblock Strings and other distributional sources in general relativity.
\newblock {\em Phys. Rev. D (3)}, 36(4):1017--1031, 1987.

\bibitem[Har02]{Hart02}
Philip Hartman.
\newblock {\em Ordinary differential equations}, volume~38 of {\em Classics in
  Applied Mathematics}.
\newblock Society for Industrial and Applied Mathematics (SIAM), Philadelphia,
  PA, 2002.
\newblock Corrected reprint of the second (1982) edition [Birkh\"{a}user,
  Boston, MA; MR0658490 (83e:34002)], With a foreword by Peter Bates.

\bibitem[HE73]{HawEll73}
Stephen W. Hawking and George F.~R. Ellis.
\newblock {\em The large scale structure of space-time}.
\newblock Cambridge University Press, London-New York, 1973.
\newblock Cambridge Monographs on Mathematical Physics, No. 1.

\bibitem[HMSSS21]{HMS:21}
Jakob Hedicke, Ettore Minguzzi, Benedict Schinnerl, Roland Steinbauer, and
  Stefan Suhr.
\newblock Causal simplicity and (maximal) null pseudoconvexity.
\newblock {\em Classical Quantum Gravity}, to appear. arXiv:2105.08998[gr-qc], 2021.

\bibitem[HW51]{HW51}
Philip Hartman and Aurel Wintner.
\newblock On the problems of geodesics in the small.
\newblock {\em Amer. J. Math.}, 73:132--148, 1951.

\bibitem[KS18]{KS18}
Michael Kunzinger and Clemens S\"{a}mann.
\newblock Lorentzian length spaces.
\newblock {\em Ann. Global Anal. Geom.}, 54(3):399--447, 2018.

\bibitem[KSSV14]{KSSV14}
Michael Kunzinger, Roland Steinbauer, Milena Stojkovi\'c, and James~A. Vickers.
\newblock A regularisation approach to causality theory for
  {$C^{1,1}$}-{L}orentzian metrics.
\newblock {\em Gen. Relativity Gravitation}, 46(8): 1738, 18, 2014.

\bibitem[KSSV15]{KSSV15}
Michael Kunzinger, Roland Steinbauer, Milena Stojkovi\'c, and James~A. Vickers.
\newblock Hawking's singularity theorem for {$C^{1,1}$}-metrics.
\newblock {\em Classical Quantum Gravity}, 32(7):075012, 19, 2015.

\bibitem[KSV15]{KSV15}
Michael Kunzinger, Roland Steinbauer, and James~A. Vickers.
\newblock The {P}enrose singularity theorem in regularity {$C^{1,1}$}.
\newblock {\em Classical Quantum Gravity}, 32(15):155010, 12, 2015.

\bibitem[Lee76]{Lee76}
C.~W. Lee.
\newblock A restriction on the topology of {C}auchy surfaces in general
  relativity.
\newblock {\em Comm. Math. Phys.}, 51(2):157--162, 1976.

\bibitem[LLS20]{LLS20}
Christian Lange, Alexander Lytchak, and Clemens S\"amann.
\newblock {Lorentz meets Lipschitz}.
\newblock  arXiv:2009.08834 [math.DG], 2020.

\bibitem[LM07]{LM07}
Philippe~G. LeFloch and Cristinel Mardare.
\newblock Definition and stability of {L}orentzian manifolds with
  distributional curvature.
\newblock {\em Port. Math. (N.S.)}, 64(4):535--573, 2007.

\bibitem[Min15]{Min15}
Ettore Minguzzi.
\newblock Convex neighborhoods for {L}ipschitz connections and sprays.
\newblock {\em Monatsh. Math.}, 177(4):569--625, 2015.

\bibitem[Min19]{Min18}
Ettore Minguzzi.
\newblock Causality theory for closed cone structures with applications.
\newblock {\em Rev. Math. Phys.}, 31(5):1930001, 139, 2019.

\bibitem[Min19a]{Min19}
Ettore Minguzzi.
\newblock Lorentzian causality theory
\newblock {\em Living Rev.\ Relativ.} 22:3, 2019.


\bibitem[Min20]{Min19a}
Ettore Minguzzi.
\newblock A gravitational collapse singularity theorem consistent with black
  hole evaporation.
\newblock {\em Lett. Math. Phys.}, 110(9):2383--2396, 2020.

\bibitem[MY82]{MSY82}
William~W. Meeks, III and Shing~Tung Yau.
\newblock The existence of embedded minimal surfaces and the problem of
  uniqueness.
\newblock {\em Math. Z.}, 179(2):151--168, 1982.

\bibitem[O'N83]{ON83}
Barrett O'Neill.
\newblock {\em Semi-{R}iemannian geometry with applications to relativity},
  volume 103 of {\em Pure and Applied Mathematics}.
\newblock Academic Press, Inc. [Harcourt Brace Jovanovich, Publishers], New
  York, 1983.

\bibitem[Sbi18]{Sbie18}
Jan Sbierski.
\newblock The {$C^0$}-inextendibility of the {S}chwarzschild spacetime and the
  spacelike diameter in {L}orentzian geometry.
\newblock {\em J. Differential Geom.}, 108(2):319--378, 2018.

\bibitem[Sbi20]{Sbie20}
Jan Sbierski.
\newblock {On holonomy singularities in general relativity and the
  $C^{0,1}_{\mathrm{loc}}$-inextendibility of spacetimes}.
\newblock arXiv:2007.12049 [gr-qc], 2020.

\bibitem[Shi93]{Shi93}
Katsuhiro Shiohama.
\newblock {\em An introduction to the geometry of {A}lexandrov spaces},
  volume~8 of {\em Lecture Notes Series}.
\newblock Seoul National University, Research Institute of Mathematics, Global
  Analysis Research Center, Seoul, 1993.

\bibitem[SS18]{SS18}
Clemens S\"{a}mann and Roland Steinbauer.
\newblock On geodesics in low regularity.
\newblock {\em J. Phys. Conf. Ser.}, 968:012010, 14, 2018.

\bibitem[SV09]{SV09}
Roland Steinbauer and James A. Vickers.
\newblock On the {G}eroch-{T}raschen class of metrics.
\newblock {\em Classical Quantum Gravity}, 26(6):065001, 19, 2009.

\bibitem[{Sä}16]{Sam16}
{Clemens Sämann}.
\newblock Global hyperbolicity for spacetimes with continuous metrics.
\newblock {\em Ann. Henri Poincar\'{e}}, 17(6):1429--1455, 2016.

\bibitem[VPE19]{VMPRE19}
Mehdi Vatandoost, Rahimeh Pourkhandani, and Neda Ebrahimi.
\newblock On null and causal pseudoconvex space-times.
\newblock {\em J. Math. Phys.}, 60(1):012502, 7, 2019.

\end{thebibliography}
\end{document}